\RequirePackage{fix-cm}
\documentclass[smallextended]{svjour3}       
\smartqed  
\usepackage{graphicx}
\usepackage{amsfonts}
\usepackage{subfloat}
\usepackage{subfigure}
\usepackage{caption}
\usepackage{color}
\usepackage{epstopdf}

\begin{document}

\title{Gaussian Variant of Freivalds' Algorithm for Efficient and Reliable Matrix Product Verification
}

\titlerunning{GVFA for Efficient and Reliable Matrix Product Verification}        

\author{Hao Ji         \and
        Michael Mascagni \and
        Yaohang Li
}


\institute{Hao Ji  \at
               Department of Computer Science \\
              Old Dominion University \\
              \email{hji@cs.odu.edu}   
           \and
           Michael Mascagni \at
             Departments of Computer Science, Mathematics, and Scientific Computing \\
             Florida State University \\
             Applied and Computational Mathematics Division\\
             National Institute of Standards and Technology\\
             \email{mascagni@fsu.edu}   
           \and
           Yaohang Li \at
              Department of Computer Science \\
              Old Dominion University \\
              Tel.: 757-683-7721\\
              Fax: 757-683-4900\\
              \email{yaohang@cs.odu.edu}           
}

\date{Received: date / Accepted: date}

\maketitle

\begin{abstract}
In this article, we consider the general problem of checking the correctness of matrix multiplication. Given three $n \times n$ matrices $A$, $B$, and $C$, the goal is to verify that $A \times B=C$ without carrying out the computationally costly operations of matrix multiplication and comparing the product $A \times B$ with $C$, term by term.  This is especially important when some or all of these matrices are very large, and when the computing environment is prone to soft errors. Here we extend Freivalds' algorithm to a Gaussian Variant of Freivalds' Algorithm (GVFA) by projecting the product $A \times B$ as well as $C$ onto a Gaussian random vector and then comparing the resulting vectors. The computational complexity of GVFA is consistent with that of Freivalds' algorithm, which is $O(n^{2})$. However, unlike Freivalds' algorithm, whose probability of a false positive is $2^{-k}$, where $k$ is the number of iterations.  Our theoretical analysis shows that when $A \times B \neq C$, GVFA produces a false positive on set of inputs of measure zero with exact arithmetic.  When we introduce round-off error and floating point arithmetic into our analysis, we can show that the larger this error, the higher the probability that GVFA avoids false positives. Moreover, by iterating GVFA $k$ times, the probability of a false positive decreases as $p^k$, where $p$ is a very small value depending on the nature of the fault on the result matrix and the arithmetic system's floating-point precision. Unlike deterministic algorithms, there do not exist any fault patterns that are completely undetectable with GVFA. Thus GVFA can be used to provide efficient fault tolerance in numerical linear algebra, and it can be efficiently implemented on modern computing architectures.  In particular, GVFA can be very efficiently implemented on architectures with hardware support for fused multiply-add operations.
\keywords{Fault-tolerance \and Algorithmic Resilience \and Gaussian Variant of Freivalds' Algorithm \and Matrix Multiplication \and Gaussian Random Vector \and Failure Probability}
\subclass{65F99 \and 65C05 \and 62P99}
\end{abstract}

\section{Introduction}
\label{section1}
As the demands on modern linear algebra applications created by the latest development of high-performance computing (HPC) architectures continues to grow, so does the likelihood that they are vulnerable to faults. Faults in computer systems are usually characterized as hard or soft, and in this article we are motivated primarily with the latter. Soft errors, defined by intermittent events that corrupt the data being processed, are among the most worrying, particularly when the computation is carried out in a low-voltage computing environment. For example, the $2,048$-node ASC Q supercomputer at Los Alamos National Laboratory reports an average of $24.0$ board-level cache tag parity errors and $27.7$ CPU failures per week \cite{michalak2005predicting}; the $131,072$-CPU BlueGene/L supercomputer at Lawrence Livermore National Laboratory experiences one soft error in its $L1$ cache every $4\textendash 6$ hours \cite{glosli2007extending}; more recently, a field study on {\it Google}'s servers reported an average of $5$ single bit errors occur in $8$ Gigabytes of RAM per hour using the top-end error rate \cite{schroeder2011dram}. The reliability of computations on HPC systems can suffer from soft errors that occur in memory, cache, as well as microprocessor logic \cite{shivakumar2002modeling}, and thus produce potentially incorrect results in a wide variety of ways. We are specifically interested in examining ways to remedy the consequences of soft errors for certain linear algebra applications.

Matrix-matrix multiplication is one of the most fundamental numerical operations in linear algebra. Many important linear algebraic algorithms, including linear solvers, least squares solvers, matrix decompositions, factorizations, subspace projections, and eigenvalue/singular values computations, rely on the casting the algorithm as a series of matrix-matrix multiplications. This is partly because matrix-matrix multiplication is one of the level-3 Basic Linear Algebra Subprograms (BLAS) \cite{dongarra1990algorithm,demmel1992stability,gallivan1987use}. Efficient implementation of the BLAS remains an important area for research, and often computer vendors spend significant resources to provide highly optimized versions of the BLAS for their machines.
Therefore, if a matrix-matrix multiplication can be carried out free of faults, the linear algebraic algorithms that spend most of their time in matrix-matrix multiplication can themselves be made substantially fault-tolerant \cite{gunnels2001fault}.  Moreover, there is considerable interest in redesigning versions of the BLAS to be more fault-tolerant, and this work will certainly contribute to that goal.

In this article, we consider the general problem of checking the correctness of matrix-matrix multiplication, i.e., given three $n \times n$ matrices $A$, $B$, and $C$, we want to verify whether $A \times B=C$. In contrast to the best known matrix-matrix multiplication algorithm running in $O(n^{2.3727})$ time \cite{coppersmith1987matrix,Williams:2012:MMF:2213977.2214056}, Freivalds' algorithm \cite{freivalds1977probabilistic} takes advantage of randomness to reduce the time to check a matrix multiplication to $O(n^2)$. The tradeoff of Freivalds' algorithm is that the probability of failure detection, a false positive, is $2^{-k}$, where $k$ is the number of iterations taken. We extend Freivalds' algorithm from using binary random vectors to floating-point vectors by projecting the $A \times B$ result as well as $C$ using Gaussian random vectors. We will refer to this algorithm as the Gaussian Variant of Freivalds' Algorithm (GVFA). By taking advantage of a nice property of the multivariate normal distribution, we show that GVFA produces a false positive on a set of random Gaussian vectors and input matrices of measure zero. Taking floating point round-off error into account, by iterating GVFA $k$ times, the probability of false positive decreases exponentially as $p^k$, where $p$ is usually a very small value related to the magnitude of the fault in the result matrix and floating-point precision of the computer architecture. We also present an efficient implementation of GVFA on computing hardware supporting fused multiplication-add operations.

The plan of the paper is the following. We first discuss two relevant algorithms from the literature for error detection in matrix-matrix multiplication. These are the Huang-Abraham scheme, discussed in section \ref{section2}, and Freivalds' algorithm, the subject of section \ref{section3}. The former is a deterministic algorithm based on carrying row and column sums along in a clever format to verify correct matrix-matrix multiplication. Freivalds' algorithm is a random projection of the computed matrix-matrix product to the same random projection of the matrix-matrix product recomputed from the original matrices using only matrix-vector multiplication. The random vector used in Freivalds' algorithm is composed of $0$'s and $1$'s. In section \ref{section4}, we present the GVFA, a variation on Freivalds' algorithm, where we instead use random Gaussian vectors as the basis of our projections. We analyze the GVFA and prove that with Gaussian vectors, a false positive occurs only on a set of Gaussian vectors of measure zero. Further analysis of false positive probabilities in the GVFA in the presence of floating-point arithmetic with round-off errors is then taken. Finally, in section \ref{section5} we provide a discussion of the results and implications for fault-tolerant linear algebraic computations and a method of enhancing the resilience of linear algebraic computations. In addition, in this final section we provide conclusions and suggest directions for future work.

\section{The Huang-Abraham Scheme and its Limit in Error Detection/ Correction}
\label{section2}
The Huang-Abraham scheme \cite{huang1984algorithm} is an algorithm-based fault tolerance method that simplifies detecting and correcting errors when carrying out matrix-matrix multiplication operations.  This is slightly different from the matrix product verification problem. The fundamental idea of the Huang-Abraham scheme is to address the fault detection and correction problem at the algorithmic level by calculating matrix checksums, encoding them as redundant data, and then redesigning the algorithm to operate on these data to produce encoded output that can be checked. Compared to traditional fault tolerant techniques, such as checkpointing \cite{bosilca2009algorithm}, the overhead of storing additional checksum data in the Huang-Abraham scheme is small, particularly when the matrices are large. Moreover, no global communication is necessary in the Huang-Abraham scheme \cite{elnozahy2002survey}. The Huang and Abraham scheme formed the basis of many subsequent detection schemes, and has been extended for use in various HPC architectures \cite{banerjee1986bounds,banerjee1990algorithm,luk1988analysis,elnozahy2002survey}.

\begin{figure}[!ht]
\centering
\subfigure[Generation of a column checksum for $A$ and a row checksum for $B$, and multiplication of the extended matrices to produce the checksum matrix for $C$]{ %
 \includegraphics[scale=0.7]{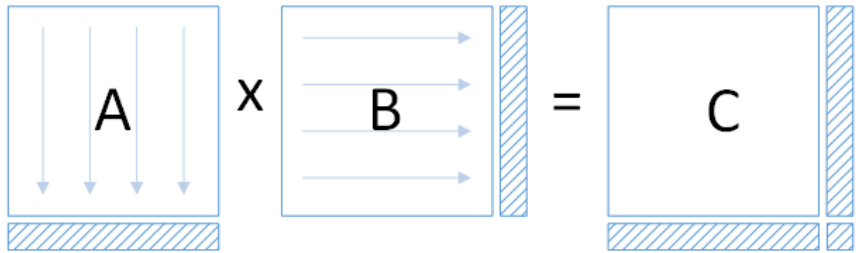}
 \label{fig:fig1a}
}

\subfigure[Mismatches in the row and column checksums indicate an element fault in the matrix product]{ %
\includegraphics[scale=0.7]{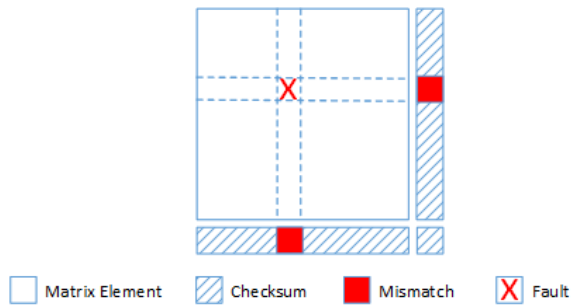}
\label{fig:fig1b}
}%

\caption{The Huang-Abraham scheme for detecting faults in matrix-matrix multiplication}
\label{fig:fig1}       
\end{figure}

Fig.~\ref{fig:fig1} illustrates the Huang-Abraham scheme \cite{huang1984algorithm} for detecting faults in matrix-matrix multiplication. First of all, column sums for $A$ and row sums for $B$ are generated and are added to an augmented representation of $A$ and $B$. These are treated as particular checksums in the subsequent multiplication. Then, multiplication of the extended matrices produces the augmented matrix for $C$ (Fig.~\ref{fig:fig1a}) where the checksums can be readily compared. Mismatches in the row and column checksums indicate an element fault in the matrix product, $C$ (Fig.~\ref{fig:fig1b}).

However, there are certain patterns of faults undetectable by the Huang-Abraham scheme. Here is a simple $2 \times 2$ example to illustrate such an undetectable pattern.

Consider the matrices
\begin{eqnarray*}
A=\left[\begin{array}{cc}2&3\\3&4\end{array}\right], B=\left[\begin{array}{cc}1&-6\\1&6\end{array}\right], \mbox{ and} \enspace  C=\left[\begin{array}{cc}5&6\\7&6\end{array}\right].
\end{eqnarray*}
Clearly $A \times B=C$ holds in this example. Then we use the Huang-Abraham scheme to calculate the column checksum for $A$ and row checksum for $B$ and we can get
$$A_{F}=\left[\begin{array}{cc}2&3\\3&4\\5&7\end{array}\right] \mbox{ and} \enspace B_{F}=\left[\begin{array}{ccc}1&-6&-5\\1&6&7\end{array}\right].$$
Then
$$A_{F} \times B_{F}=\left[\begin{array}{ccc}5&6&11\\7&6&13\\12&12&24\end{array}\right]=C_{F}.$$
However, if there is a fault during the computation of $C$ which causes an exchange of the first and second columns, an erroneous result matrix $C'=\left[\begin{array}{cc}6&5\\6&7\end{array}\right]$ is generated by exchanging the columns of $C$. Column or row exchange, usually caused by address decoding faults \cite{van1991testing}, is a commonly observed memory fault pattern \cite{cheng2003fame}. The problem is that the checksum matrix of $C'$ becomes ${C'}_{F}=\left[\begin{array}{ccc}6&5&11\\6&7&13\\12&12&24\end{array}\right]$, where both the row and column checksums match those of the true product of $A \times B$. Consequently, the Huang-Abraham scheme fails to detect this fault.

The Huang-Abraham scheme can be viewed as a linear constraint satisfaction problem (CSP), where the variables are the $n^2$ entries in the product matrix, $C$, the constraints are the $2n$ row and column checksums.  Also, the $2n \times n^{2}$ coefficient matrix in the under-determined linear CSP system equation specifies the selection of row or column elements, as shown in Fig. \ref{fig:fig2}. Clearly, a product matrix, $C$, that does not satisfy the CSP equations indicates errors in $C$ detectable by the Huang-Abraham scheme. The unique, correct product matrix, $C$, satisfies the CSP equations. Nevertheless, other possible product matrices satisfying the CSP equations are the fault patterns undetectable by the Huang-Abraham scheme. Only when at least $n^2$ constraints with different element selection are incorporated so that the rank of the coefficient matrix in the CSP equation is $n^2$, can the undetectable fault patterns be eliminated. However, this situation is equivalent to simply checking every element in $C$.


\begin{figure}[!ht]
  \includegraphics[scale=0.4]{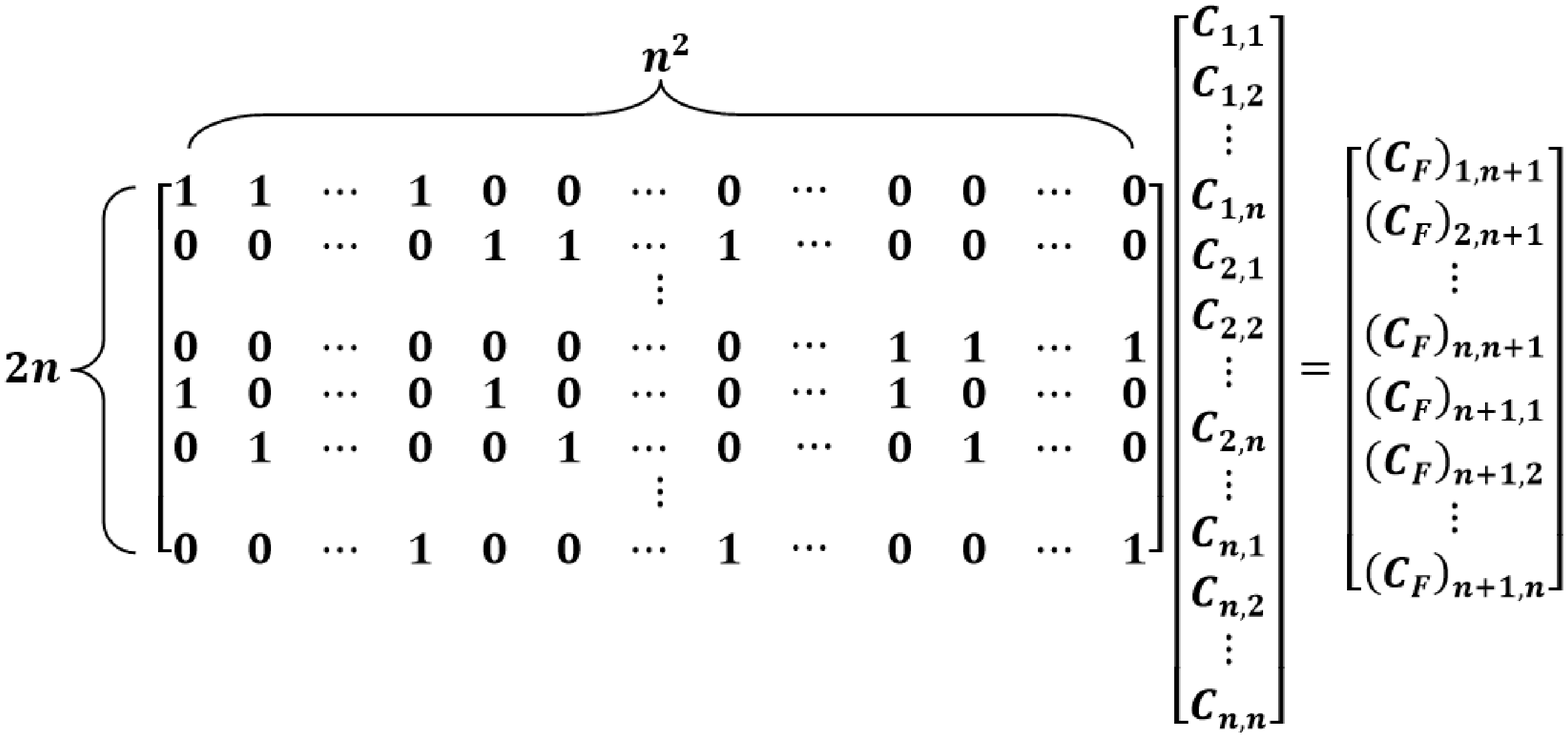}
\caption{Under-determined CSP system in the Huang-Abraham Scheme}
\label{fig:fig2}       
\end{figure}

It is important to notice that there are an infinite number of existing fault patterns that satisfy the checksum constraints and thus are undetectable by the Huang-Abraham scheme, even in the above simple $2 \times 2$ example (the rank of the CSP coefficient matrix is $3$). Moreover, as dimension, $n$, increases, the number of checksum constraints increases only linearly but the number of elements in a matrix has quadratic growth. Therefore, the undetectable patterns in the Huang-Abraham scheme increase quadratically with $n$. As a result, for multiplications in large matrices, fault detection methods based on the Huang-Abraham scheme can generate false positive results for a large number of circumstances.

\section{Freivalds' Algorithm}
\label{section3}
The fault detection methods based on the Huang-Abraham scheme are deterministic algorithms. As many randomized fault tolerance algorithms \cite{Li02grid-basedmonte,li2003analysis}, with the tradeoff of random uncertainty, Freivalds \cite{freivalds1977probabilistic} showed that a probabilistic machine can verify the correctness of a matrix product faster than direct recalculation. The procedure of the corresponding method, later named Freivalds' algorithm, is described in Algorithm 1.

\begin{table}[!ht]
\label{tab:alg1}
\begin{center}
\begin{tabular}{|p{10cm}|}
\hline
\multicolumn{1}{|c|}{Algorithm 1: Freivalds' Algorithm}\\

\item[1.]     Randomly sample a vector $\omega \in\{0,1\}^n$ with $p =\frac{1}{2}$ of $0$ or $1$. \\
\item[2.]     Calculate the projection of $C$ onto $\omega$: $C\omega=C \times \omega$. \\
\item[3.] Calculate the projection of the product $A\times B$ onto $\omega$: $AB\omega=A\times (B \times \omega)$. \\
\hline
\end{tabular}
\end{center}
\end{table}

Obviously, if $A \times B=C$, $C \omega = AB\omega$ always holds. Freivalds proved that when $A \times B \neq C$, the probability of $C\omega=AB\omega$ is less than or equal to $\frac{1}{2}$. The running time of the above procedure is $O(n^2)$ with an implied multiplier of $3$, as it is comprised of three matrix-vector multiplications.  This is an upper bound as one can perhaps optimize the evaluation of $B\omega$ and $C\omega$. By iterating the Freivalds' algorithm $k$ times, the running time becomes $O(kn^2)$ and the probability of a false positive becomes less than or equal to $2^{-k}$, according to the one-sided error. More generalized forms of Freivalds' algorithm have also been developed, mainly based on using different sampling spaces \cite{chinn1993bounds,alon1990simple,naor1993small,lisboa2007low}. Given at most $p$ erroneous entries in the resulted matrix product, Gasieniec, Levcopoulos, and Lingas extended Freivalds' algorithm to one with correcting capability running in $O(\sqrt{p}n^2\log(n)\log(p))$ time \cite{gkasieniec2014efficiently}.

\section{A Gaussian Variant of Freivalds' Algorithm (GVFA)}
\label{section4}
\subsection{Extending Freivalds' Algorithm using Gaussian Vectors}
\label{subsection4.1}
Freivalds' original algorithm, and most of its extensions are based on integer matrices or matrices over a ring and sampling from discrete spaces. Clearly, we can also apply Freivalds' algorithm to matrices with real or complex entries with the random vector remaining zeros and ones. A simple extension is to project $A \times B$ and $C$ onto a vector $\omega_P$ of form $\omega_P=(1, r, r^2, ..., r^{n-1})^T$, where $r$ is a random real number. A false positive occurs only when $r$ is the root of the corresponding polynomial. However, in practice, $r^{n-1}$ can easily grow too large or small exceeding floating point representation \cite{korec2014deterministic}. 

Here we also extend Freivalds' algorithm by using Gaussian random vectors for the projection. We use the fact that the multivariate normal distribution has several nice properties \cite{muirhead2009aspects}, which have been used for detecting statistical errors in distributed Monte Carlo computations \cite{li2003analysis}. The extended algorithm is described in Algorithm 2.

\begin{table}[!ht]
\label{tab:alg2}
\begin{center}
\begin{tabular}{|p{10cm}|}
\hline
\multicolumn{1}{|c|}{Algorithm 2: Gaussian Variant of Freivalds' Algorithm}\\

\item[1.] Generate a Gaussian random vector, $\omega_G$, made up of $n$ independent (but not necessarily identically) distributed normal random variables with finite mean and variance.\\
\item[2.] Calculate the projection of C on $\omega_G$: $C\omega_G=C \times \omega_G$.\\
\item[3.] Calculate the projection of product $A \times B$ on $\omega_G$: $AB\omega_G=A \times (B \times \omega_G)$. \\
\hline
\end{tabular}
\end{center}
\end{table}

This algorithm, which we call a Gaussian variant of Freivalds algorithm (GVFA), requires three matrix-vector multiplications and only one vector comparison for fault detection.

\subsection{Theoretical Justification}
\label{subsection4.2}
Similar to Freivalds' algorithm, in GVFA if $A \times B=C$, $C\omega_G=AB\omega_G$ always holds within a certain floating point round-off threshold. When $A \times B \neq C$, the chance that $C\omega_G=AB\omega_G$ is a false positive event occurs with measure zero in exact arithmetic, as shown in Theorem \ref{theorem1}.
We first state a result of Lukacs and King \cite{lukacs1954property}, shown as Proposition \ref{proposition1}, which will be used in the proof of Theorem \ref{theorem1}. The main assumption of Proposition \ref{proposition1} is the existence of the $n$th moment of each random variable, which many distributions, particularly the normal distribution, have. One important exception of the normal is that it is the limiting distribution for properly normalized sums of random variables with two finite moments. This is Lindeberg's version of the Central Limit Theorem \cite{lindeberg1922neue}. 

\vspace*{\baselineskip}

\begin{proposition}
\label{proposition1}
Let $X_1, X_2, \cdots,X_n$ be $n$ independent (but not necessarily identically) distributed random variables with variances $\sigma_{i}^2$, and assume that the $n$th moment of each $X_{i}$  $\left( i=1,2,\cdots,n \right)$ exists and is finite. The necessary and sufficient conditions for the existence of two statistically independent linear forms $Y_1=\sum_{i=1}^{n}{a_{i}X_{i}}$ and $Y_2=\sum_{i=1}^{n}{b_{i}X_{i}}$ are
\renewcommand\labelenumi{(\theenumi)}
\begin{enumerate}
\item Each random variable which has a nonzero coefficient in both forms is normally distributed.
\item $\sum_{i=1}^{n}{a_{i}b_{i} \sigma_{i}^2}=0$.
\end{enumerate}
\end{proposition}

\vspace*{\baselineskip}

\begin{theorem}
\label{theorem1}
If $A \times B \neq C$, the set of Gaussian vectors where $C\omega_G=AB\omega_G$ holds in Algorithm 2 has measure zero.
\end{theorem}
\begin{proof}
Let the matrix $ \Delta \in\mathbb{R}^{n \times n}$ denote $ AB-C$.  Since $ A \times B \neq C$, $rank(\Delta)= r > 0$, and $dim(null(\Delta))=n-rank(\Delta)=n-r<n$. Here $dim(\cdot)$ denotes dimension and $null(\cdot)$ denotes the null space, i.e., 
$null (\Delta) = \{ x \in\mathbb{R}^{n}  : \Delta \times x=0 \}$.

We can now find $n - r$ of orthonormal vectors, $v_1, v_2, \cdots, v_{n-r}$, to form a basis for $null(\Delta)$, such that
$$null(\Delta)=span\{v_1,v_2,\cdots,v_{n-r} \},$$
and $r$ more orthonormal vectors, $v_{n-r+1},v_{n-r+2},\cdots,v_n$, such that
$$\mathbb{R}^{n}=span\{v_1,v_2,\cdots,v_{n-r},v_{n-r+1},v_{n-r+2},\cdots,v_n \}.$$

Any vector, and in particular the Gaussian vector, $\omega_{G}$ can be written in this basis as
$$\omega_G=\sum_{i=1}^{n}{\delta_i v_i},$$
where $\delta_i$ are the weights in this particular orthonormal coordinate system. If we denote $V= \left[ v_1, v_2, \cdots, v_{n-r}, v_{n-r+1}, v_{n-r+2}, \cdots,v_n \right]$ , we have
$$V\omega_G=\left[\delta_1,\delta_2,\cdots,\delta_{n-r},\delta_{n-r+1},\delta_{n-r+2},\cdots,\delta_n \right].$$
$C\omega_G=AB\omega_G$ holds in Algorithm 2 only if $A(B\omega_G)-C\omega_G=(AB-C) \omega_G=\Delta\omega_G=0$. This means that $\omega_G\in null(\Delta)$, i.e., $\delta_{n-r+1} =0, \delta_{n-r+2} = 0,\cdots,\delta_n=0$.
Due to the fact that $\omega_G$ is a Gaussian random vector and $V$ is an orthogonal matrix, Proposition \ref{proposition1} tells us that the elements, $\delta_i$, in the resulting vector $V\omega_G$ are normally distributed and statistically independent. With a continuous probability distribution, the discrete event where $\delta_i=0$ for all $i>n-r$ occurs on a set of measure zero and we will say here that it has probability zero. Hence, GVFA using a Gaussian random projection will have unmatched $C\omega_G$ and $AB\omega_G$ when $A \times B \neq C$ on all but a set of measure zero of Gaussian vectors, which we will say is probability one. \qed
\end{proof}

This argument in Theorem \ref{theorem1} is rather direct, but we must point out that the arguments are true when the computations are exact. In next subsection, we will analyze GVFA when float-point errors are present.

\subsection{Practical Use in Floating-Point Matrix Product Verification}
\label{subsection4.3}
In computer implementations of arithmetic with real numbers, one commonly uses floating-point numbers and floating-point arithmetic.  Floating-point numbers are represented as finite numbers in the sense that they have a fixed mantissa and exponent size in number of bits. Therefore, there will be a small probability, $p$, that $C\omega_G=AB\omega_G$ still holds due to unfortunate floating-point operations in a system with a known machine epsilon, $\epsilon$, when $A \times B \neq C$. The value of $p$ depends on the magnitude of the error between $A \times B$ and $C$ as well as $\epsilon$, whose upper bound is justified in Theorem \ref{theorem2}.

\vspace*{\baselineskip}

\begin{theorem}
\label{theorem2}
Assume that $\omega_G$ is a standard Gaussian random vector, whose elements are i.i.d.~normal variables with mean $0$ and variance $1$, i.e.,~the standard normal. Let $\Delta=A \times B-C$, then the probability, $p$, that $C\omega_G=AB\omega_G$ holds in Algorithm 2 using a standard Gaussian random vector $\omega_G$ under floating-point uncertainty of size $\epsilon$ is
$$p\leq 2 \Phi \left( \left|\frac{\epsilon}{\widetilde{\sigma}} \right| \right)-1,$$ where $\Phi(\cdot)$ is the cumulative density function of the standard normal, and $\widetilde{\sigma}$ is a constant only related to $\Delta$.
\end{theorem}
\begin{proof}
Since $A \times B \neq C$,   $\Delta=A \times B-C \neq 0$.
Consider the $i$th element, $g_i$, of the product vector $g=\Delta \times \omega_G$, we have
$$g_i=(\Delta \times \omega_G )_i=\sum_{j=1}^n{\Delta_{ij} (\omega_G )_j}.$$
Given $\epsilon$, only if $|g_i|\leq \epsilon$ for all $i=1,\cdots,n,$ can $C\omega_G=AB\omega_G$ hold. Since $\omega_G$ is a standard normal random vector, the $g_i$ for all $i=1,\cdots,n$, are normally distributed as well.  This is because they are linear combinations of normals themselves.  The key is to compute what the mean and variance is of the $g_i$.  The components of the $\omega_G$ are i.i.d.~standard normals. Thus we have that $E \left[ (\omega_G)_j \right] =0$ and $E \left[ (\omega_G)_j^2 \right] = 1$, for all $j=1,\cdots,n.$  Also, we have that $E \left[ (\omega_G )_i (\omega_G )_j \right]=0$ when $i \ne j$.  This allows us to compute the mean:
$$E(g_i )=E \left[ \sum_{j=1}^n{\Delta_{ij} (\omega_G )_j}\right]= \sum_{j=1}^n{\Delta_{ij} E \left[ (\omega_G )_j \right]}=0,$$
and the second moment about the mean, i.e.,~the variance:
\begin{eqnarray*}
E \left[g_i^2 - E(g_i)^2 \right] &=& E \left[ g_i^2 \right]
= E \left[\sum_{j=1}^n{\Delta_{ij} (\omega_G )_j}\right]^2 \\
&=& E \left[\sum_{j=1}^n \Delta_{ij}^2 \times 1 \right]
= \sum_{j=1}^n \Delta_{ij}^2.
\end{eqnarray*}
So we have that the $g_i$'s are normally distributed with mean zero and variance $\widetilde{\sigma}_{i}^{2}=\sum_{j=1}^n \Delta_{ij}^2$, i.e.,~$g_i \sim N \left(0,\widetilde{\sigma}_{i}^{2} \right)$.

Then, the probability that $|g_i|\leq\epsilon$  can be computed as follows.  Since $g_i\sim N\left( 0,\widetilde{\sigma}_{i}^{2}\right)$, we know that $\frac{g_i}{\widetilde{\sigma}_{i}^{2}} \sim N(0,1)$, and so we define the new variables  $\widetilde{g}_{i}=  \frac{g_i}{\widetilde{\sigma}_{i}^{2}} $ and $\widetilde{\epsilon}  = \frac{\epsilon}{\widetilde{\sigma}_{i}^{2}}$, and so we have
\begin{eqnarray*}
p\left(|g_i|\leq \epsilon \right) &=& p\left(-\epsilon\leq g_i \leq\epsilon \right) \\
&=& p \left( -\widetilde{\epsilon}  \leq  \widetilde{g}_{i}  \leq \widetilde{\epsilon} \right) \\
&=& \int_{- \widetilde{\epsilon}}^{\widetilde{\epsilon}} \frac{1}{\sqrt{2\pi}} e^{-\frac{1}{2} t^2}dt \\
&=& \Phi(\widetilde{\epsilon} )-\Phi(-\widetilde{\epsilon} ).
\end{eqnarray*}
Since the probability density function of a standard normal is an even function, we have that $\Phi(\widetilde{\epsilon} )+\Phi(-\widetilde{\epsilon} )=1$, and so we can use $- \Phi(-\widetilde{\epsilon} )=\Phi(\widetilde{\epsilon} )-1$ to get:
$$p(-\epsilon\leq g_i\leq \epsilon) =2\Phi(\widetilde{\epsilon} )-1=2\Phi\left(\left| \frac{\epsilon}{\widetilde{\sigma}_i} \right|\right)-1.$$

Now let us consider computing an upper bound on $p(|g_i |\leq\epsilon ,i=1,\cdots,n)$.  We have proven that the $g_i$'s are normal random variables, but they are not necessarily independent.  And so for this we use some simple ideas from conditional probability.  By example, consider
$$
p(|g_1| \leq \epsilon \mbox{ and } |g_2| \leq \epsilon) =
p(|g_2| \leq \epsilon \mbox{ given } |g_1| \leq \epsilon)p(|g_1| \leq \epsilon) \leq p(|g_1| \leq \epsilon).
$$
The inequality holds due to the fact that the probabilities are numbers less than one.  Now consider our goal of bounding
$$
p(|g_i|\leq\epsilon ,i=1,\cdots,n) \leq p(|g_1| \leq \epsilon) = \left[2\Phi\left(\left| \frac{\epsilon}{\widetilde{\sigma}_1} \right|\right)-1 \right],
$$
by iterating the conditional probability argument $n$ times.  By reordering we could have chosen the bound utilizing any of the $g_i$'s.  However, let us define $\widetilde{\sigma}=\max_i \sqrt{\sum_{j=1}^n {\Delta_{ij}^{2}}} $, i.e.,~the maximal standard deviation over all the $g_i$'s, which is only related to the matrix $\Delta$.  We can use that value instead to get
$$
p=p(|g_i|\leq\epsilon ,i=1,\cdots,n) \leq  \left[2\Phi\left(\left| \frac{\epsilon}{\widetilde{\sigma}} \right|\right)-1 \right].
$$ \qed
\end{proof}

As an interesting corollary, we can get a better bound in the case the at the $g_i$'s are independent.  In that case
$$p(|g_i |\leq\epsilon ,i=1,\cdots,n) = \prod_{i=1}^n p(|g_i |\leq\epsilon ) = \prod_{i=1}^n \left[2\Phi\left(\left| \frac{\epsilon}{\widetilde{\sigma}_i} \right|\right)-1 \right].$$
Let $\widetilde{\sigma}=\max_i \sqrt{\sum_{j=1}^n{\Delta_{ij}^{2}}}$, i.e.,~the maximal standard deviation over all the $g_i$'s, which is only related to the matrix $\Delta$. Hence for all $i=1,\cdots,n,$ we have that $$2\Phi\left(\left| \frac{\epsilon}{\widetilde{\sigma_i}} \right|\right)-1 \le 2\Phi\left(\left| \frac{\epsilon}{\widetilde{\sigma}} \right|\right)-1.$$
And so, finally we get that

\begin{eqnarray*}
p&=&p(|g_i |\leq\epsilon ,i=1,\cdots,n) \\
&=& \prod_{i=1}^n \left[ 2\Phi\left(\left| \frac{\epsilon}{\widetilde{\sigma}_i} \right|\right)-1 \right] \\
&\le&
\left[ 2\Phi\left(\left|\frac{\epsilon}{\widetilde{\sigma}} \right|\right)-1 \right]^n  \\
&\le& 2\Phi\left(\left|\frac{\epsilon}{\widetilde{\sigma}} \right|\right)-1.
\end{eqnarray*}
The last inequality is true since the number raised to the $n$th power is less than one.

Note, that independence gives probability of a false positive that is $n$ times smaller than in the general, dependent case.  The conclusion of this seems to be that the bound in the dependent case is overly pessimistic, and we suspect that in cases where the matrix $\Delta$ is very sparse, due to a very small number of errors, that we are in the independent $g_i$'s case or have very little dependence, and these more optimistic bounds reflect what happens, computationally.

Theorem \ref{theorem2} reveals two interesting facts about GVFA in term of practical floating-point matrix product verification:
\renewcommand\labelenumi{(\theenumi)}
\begin{enumerate}
\item The bigger the error caused by the fault, the higher the probability that it can be captured. $p$ is usually very small because the floating point bound, $\epsilon$, is very small.
\item Similar to the original Freivalds' algorithm, higher confidence can be obtained by iterating the algorithm multiple times. In fact, if we iterate $k$ times using independent Gaussian random vectors, the probability of false positive decreases exponentially as $p^k$. Actually, due to the fact that $p$ is usually very small, one or a very small number of iterations will produce verification with sufficiently high confidence.
\end{enumerate}


One comment that should be made is that if we consider $\int_{-\widetilde{\epsilon}}^{\widetilde{\epsilon}}\frac{1}{\sqrt{2\pi}} e^{-\frac{1}{2} t^2} dt$ when $\widetilde{\epsilon}$ is small, we can easily approximate this. Since the integrand is at its maximum at zero, and is a very smooth function, analytic actually, this integral is approximately the value of the integrand at zero times the length of the integration interval, i.e.,~ $\int_{-\widetilde{\epsilon} }^{\widetilde{\epsilon}} \frac{1}{\sqrt{2\pi}} e^{-\frac{1}{2} t^2} dt \leq 2\widetilde{\epsilon} \frac{1}{\sqrt{2\pi}} =\widetilde{\epsilon} \sqrt{\frac{2}{\pi}}$.  This is justified as $\widetilde{\epsilon}$ is a number on the order of the machine epsilon, which is $2^{-23}$ in single precision or $2^{-52}$ in double precision floating point, divided by $\widetilde{\sigma}_{i}^{2}=\sum_{j=1}^n \Delta_{ij}^2$.

Compared to deterministic methods, such as the Huang-Abraham scheme, GVFA has the following advantages:
\renewcommand\labelenumi{(\theenumi)}
\begin{enumerate}
\item Certain fault patterns, as shown in Section \ref{section2}, are undetectable in deterministic methods such as the Huang-Abraham scheme. Deterministic methods absolutely cannot detect faults with certain patterns, i.e.,~certain patterns are detected with probability zero. In contrast, there are no fault patterns that are undetectable by GVFA with 100\% probability. Moreover, iterating the algorithm multiple times can increase the probability of detecting any fault pattern any value less than one by iteration.
\item From the computational point-of-view, normal random vectors are generated independently of $A$, $B$, and $C$, which avoids the costly computation of checksums.
\end{enumerate}

\subsection{Huang-Abraham-like GVFA}
\label{subsection4.4}
GVFA can also be implemented in a way similar to that of the Huang-Abraham scheme by providing row and column verification, as shown in Algorithm 3.

\begin{table}[!ht]
\label{tab:alg3}
\begin{center}
\begin{tabular}{|p{10cm}|}
\hline
\multicolumn{1}{|c|}{Algorithm 3: Huang-Abraham-like GVFA} \\

\item[1.]  Generate two $n$-dimensional Gaussian random vectors, $\omega_R$, a column vector, and $\omega_C$, a row vector, where they independent (but not necessarily identically) distributed normal random variables with finite mean and variance. \\

\item[2.]  Calculate the projection of $C$ on $\omega_R$ and $\omega_C$: $\omega_R C=\omega_R \times C$ and $C\omega_C=C \times \omega_C$. \\

\item[3.]  Calculate the projection of the product $A \times B$ on $\omega_R$ and $\omega_C$: $\omega_R AB=(\omega_R \times A) \times B$ and $AB\omega_C=A \times (B \times \omega_C)$. \\

\hline
\end{tabular}
\end{center}
\end{table}

Similar to the Huang-Abraham scheme, a mismatched element of the row vectors of $\omega_R C$ and $\omega_R AB$ as well as that of the column vectors of $C\omega_C$ and $AB\omega_C$ uniquely identify a faulty element in $C$. By considering floating-point errors, the false positive probability of identifying this fault becomes $p^2$, according to the analysis in Section \ref{subsection4.3}. However, the computational cost doubles with six matrix-vector multiplications and two vector comparisons.  This is essentially the same work as doing two independent iterations of the GVFA, and obtains the same bound.

\subsection{Implementation using Fused Multiply-Add Hardware}
\label{subsection4.5}
The Fused Multiply-Add (FMA) machine instruction performs one multiply operation and one add operation with a single rounding step \cite{hokenek1990second}.  This was implemented to enable potentially faster performance in calculating the floating-point accumulation of products, $a: =a+b \times c$. Recall that the GVFA employs three matrix-vector multiplications to project $A \times B$ and $C$ onto a normal random vector, which requires a sequence of product accumulations that cost $3n(2n-1)$ floating-point operations. Therefore, the performance of the GVFA can be potentially boosted on modern computing architectures that support the FMA. More importantly, due to a single rounding step used in the FMA instruction instead of two roundings within separate instructions, less loss of accuracy occurs when using the FMA instruction in calculating the accumulation of products \cite{boldo2011exact}.  This should further reduce the floating-point rounding errors that cause false positives.

\section{Discussion and Conclusions}
\label{section5}
In this paper, we extend Freivalds' algorithm, which we call the Gaussian variant of Freivalds' algorithm (GVFA), to the real domain by random projection using vectors whose coefficients are i.i.d. normal random variables. If $A \times B \neq C$, the probability that the resulting vectors match is zero using exact arithmetic. Considering the round-off errors in floating-point operations, the probability of fault detection depends on the magnitude of the error caused by the fault as well as the floating point precision. The new GVFA can be iterated $k$ times with the probability of false positives decreasing exponentially in $k$. In addition to matrix-matrix multiplication, the new algorithm can be applied to verify a wide variety of computations relevant to numerical linear algebra as it provides fault tolerance to the computation that defines level 3 of the BLAS. GVFA can also be used to enforce the trustworthiness of outsourcing matrix computations on untrusted distributed computing infrastructures such as clouds or volunteer peer-to-peer platforms \cite{lei2014cloud,kumar:hal-00876156}.

The GVFA can be easily extended to a more general matrix multiplication operation where $A$ is $m \times p$, $B$ is $p \times n$, and $C$ is $m \times n$. The overall computational time then becomes $O(mp+np)$. The algorithm can be further extended to verify the product of $N$ matrices, which requires overall $N+1$ matrix-vector multiplications.  The GVFA can also be applied to verifying a wide variety of matrix decomposition operations such as LU, QR, Cholesky, as well as eigenvalue computations, and singular value decompositions. In this case, faults are not in the product matrix but occur in the decomposed ones instead. Anyway, the GVFA can be directly applied with no modifications necessary.

The GVFA is a new tool to detect faults in numerical linear algebra, and since it is based on random Gaussian projection, it is related to the many new randomized algorithms being used directly in numerical linear algebra \cite{halko2011finding,drineas2006fast1,drineas2006fast2,drineas2006fast3,eriksson2011importance}.
The fundamental idea of these randomized algorithms is to apply efficient sampling on the potentially very large matrices to extract their important characteristics so as to fast approximate numerical linear algebra operations. We believe that the GVFA will be a very useful tool in the development of fault-tolerant and otherwise resilient algorithms for solving large numerical linear algebra problems.  In fact, it seems that the GVFA's similarity to other, new, stochastic techniques in numerical linear algebra affords the possibility of creating stochastic linear solvers that are by their very nature resilient and fault-tolerant.  This is highly relevant for new machines being developed in HPC to have maximal floating-point operations per second (FLOPS) while existing within restrictive energy budgets. These HPC systems will be operating at voltages lower than most current systems, and so they are expected to be particularly susceptible to soft errors. However, even if one is not anticipating the use of these high-end machines, the trend in processor design is to lower power, and is being driven by the explosion of mobile computing. Thus, the ability to reliably perform complicated numerical linear algebraic computations on systems more apt to experience soft faults is a very general concern. The GVFA will make it much easier to perform such computations with high fidelity in HPC, cloud computing, mobile applications, as well in big-data settings.

\begin{acknowledgements}
We would like to thank Dr. Stephan Olariu for his valuable  suggestions on the manuscript. This work is partially supported by National Science Foundation grant 1066471 for Yaohang Li and Hao Ji acknowledges support from an ODU Modeling and Simulation Fellowship.  Michael Mascagni's contribution to this paper was partially supported by National Institute of Standards and Technology (NIST) during his sabbatical.

The mention of any commercial product or service in this paper does not imply an endorsement by NIST or the Department of Commerce.
\end{acknowledgements}

\bibliographystyle{spmpsci}
\bibliography{mybibfile}


\end{document}